\documentclass[11pt]{llncs}
\usepackage{qtree}

\usepackage{enumerate}

\newcommand*{\nat}{\ensuremath{\mathbb{N}}}

\usepackage{latexsym}

\newcommand{\To}{\longrightarrow}

\newcommand{\manyk}[1]{{#1}_1,        \ldots       ,       {#1}_{k}}

\newcommand{\many}[2]{{#1}_1,        \ldots       ,       {#1}_{#2}}

\renewcommand{\k}{{\sf K}}
\renewcommand{\d}{{\sf D}}

\newcommand{\df}{{\sf D4}}

\newcommand{\jd}{{\sf JD}}
\newcommand{\jdf}{{\sf JD4}}

\newcommand{\lp}{{\sf LP}}

\usepackage{complexity}

\newcommand{\term}[2]{#1  \colon_{#2} \!}

\usepackage{mathpartir}
\newcommand{\cs}{\mathcal{CS}}

\renewcommand{\M}{\mathcal{M}}
\newcommand{\T}{\mathcal{T}}
\newcommand{\F}{\mathcal{F}}

\newcommand{\V}{\mathcal{V}}
\newcommand{\J}{\mathcal{J}}

\usepackage{amsmath}
\usepackage{amsfonts}
\usepackage{amssymb}

\usepackage{flexisym}
\usepackage{breqn}

\renewcommand{\A}{{\mathcal{A}}}

\title{Complexity Jumps In Multiagent Justification Logic Under Interacting Justifications}
\author{Antonis Achilleos} 
\institute{The Graduate Center of The City University of New York, \\
	365 Fifth Avenue
	New York, NY 10016 USA	
	\\ \email{aachilleos@gc.cuny.edu}}
	
\begin{document}
		
		\maketitle
		
		\begin{abstract} 
			The Logic of Proofs, LP, and its successor, Justification Logic, is a refinement of the modal logic approach to epistemology in which proofs/justifications are taken into account. In 2000 Kuznets showed that satisfiability for {\lp} is in the second level of the polynomial hierarchy, a result which has been successfully repeated for all other one-agent justification logics whose complexity is known.
			
			We introduce a family of multi-agent justification logics with interactions between the agents' justifications, by extending and generalizing the two-agent versions of the Logic of Proofs introduced by Yavorskaya in 2008. Known concepts and tools from the single-agent justification setting are adjusted for this multiple agent case. We present tableau rules and some preliminary complexity results. In several cases the satisfiability problem for these logics remains in the second level of the polynomial hierarchy, while for others it is {\PSPACE} or \EXP-hard. Furthermore,  this problem becomes {\PSPACE}-hard even for certain two-agent logics, while there are $\EXP$-hard logics of three agents.
		\end{abstract}
		
		\section{Introduction}
		
		Justification Logic is a family of logics of justified beliefs. Where epistemic modal logic treats formulas of the form $K \phi$ with the intended meaning that an agent knows/believes $\phi$, in Justification Logic we consider formulas of the form $\term{t}{}\phi$ with the intended meaning that $t$ is a justification for $\phi$ - or that the agent has justification $t$ for $\phi$. 
		The first Justification Logic was \lp, the logic of proofs, and appeared in \cite{Art95TR} by Artemov, but it has since developed in a wide system of explicit epistemic logics with notable complexity properties that significantly differ from the corresponding modal logics: while every single-agent justification logic whose complexity has been studied has its derivability problem in $\Pi_2^p$ (the second level of the polynomial hierarchy), the corresponding modal logics have \PSPACE-complete derivability problems. Furthermore certain significant fragments of these justification logics have an even lower complexity- \NP, or even \P\ in some cases. For an overview of Justification Logic see \cite{Art08RSL}. For an overview of complexity results of (single-agent) Justification Logic, see \cite{Kuz08PhD}.
		
		In epistemic situations we often have multiple agents, so as it is the case with modal logic, there is a need for a multi-agent justification logic.
		In \cite{multi2}, Yavorskaya presents two-agent variations of \lp. These logics feature interactions between the two agents' justifications: for $\lp_\uparrow$, for instance, every justification for agent 1 can be converted to a justification for agent 2 for the same fact and we have the axiom $\term{t}{1}\phi \rightarrow \term{\uparrow\!\! t}{2} \phi$,\footnote{We take some liberty with the notation to keep it in line with this paper.} while $\lp_!$ comes with the extra axiom $\term{t}{1} \phi \rightarrow \term{!t}{2}\term{t}{1}$, so agent 2 is aware of agent 1's justifications.
		
		In \cite{Achilleos14TRtwoagent}, we extended Yavorskaya's logics to two-agent variations of other justification logics, as well as to combinations of two different justification logics. 
		We then gave tableau procedures to prove that most of these logics were in the second step of the polynomial hierarchy, an expected result which mimics the ones for single-agent justification logics from \cite{DBLP:conf/csl/Kuznets00,Kuz08PhD,Achilleos:wollic11}. For some cases, however, we were able to prove \PSPACE-completeness, which was a new phenomenon for Justification Logic.
		
		In this paper we continue our work from \cite{Achilleos14TRtwoagent}. 
		We provide a general family of multi-agent logics. Each member of this family we call $(J^{n}_{D,F,V,C})_\cs$, where $n, D, F, V, C$ are parameters of the logic. For $(J^{n}_{D,F,V,C})_\cs$ we consider $n$ agents and the interactions between the agents' justifications are described by binary relations on agents, $V$ and $C$. Furthermore, not all agents are equally reliable: $D$ and $F$ are sets of agents, all agents in $D$ have consistent beliefs and all agents in $F$ have true beliefs. These concepts are made precise in section \ref{multidefinitions}. It is our goal to provide a flexible system capable of modelling situations of many diverse agents, or diverse types of justifications, allowing for reasonably general interactions among their justifications.
		
	For this family of logics we provide semantics and a general tableau procedure and then we make  observations on the complexity of the derivation problem for its members. In particular, we demonstrate that all logics in this family have their satisfiability problem in {\NEXP} - under reasonable assumptions. 
%	Furthermore, this is an optimal result, as some logics are \NEXP-complete. 
	This family demonstrates significant variety, as it also includes \PSPACE- and \EXP-complete members, while of course some of its members have their satisfiability problem in $\Sigma_2^p$. This is a somewhat surprising result, as all single-agent justification logics whose complexity is known have their satisfiability problem in $\Sigma_2^p$.

%		For this we generalize the known results for the single-agent cases and demonstrate new phenomena that arise in the multi-agent context and that affect the complexity of the derivation problem in a significant way.
		
%		Other multi-agent justification logics have already been introduced (for example, see \cite{BucKuzStu11JANCL,Ren11SynLib}), but they present a different approach. This paper focuses on interactions  between the agents' justifications and not on any actual interaction between the agents themselves. 
		
		This paper is organised as follows. In section \ref{multidefinitions}, we give the base definitions of the syntax, axioms and semantics for each logic in the family. 
%		We prove soundness and completeness using similar methods as the ones used for the single-agent logics. While doing so, we try to provide a more unified version of M-models (by Mkrtychev in \cite{Mkr97LFCS}) and F-models (by Fitting in \cite{Fit05APAL}) semantics. 
%		At the end of the section, we give some examples of multi-agent justification logics. 
		Then we reintroduce the star calculus, an invaluable tool, and our first complexity results, mirroring the ones for single-agent justification logic (see \cite{NKru06TCS}). The version of the star calculus we provide is somewhat different than the usual ones in that it is based on a given frame. If the frame includes a single world, we get the usual, more familiar version. In section \ref{tableaux} we give general tableau rules for each of our logics. Naturally, the rules are parameterized by the logic's parameters, including the interactions between the agents, so special attention is given in that section to these interactions. We then go on and further optimize the tableau procedure with respect to the number of world-prefixes it produces; this results in a $\Sigma_2^p$ upper bound for the satisfiability of a general class of logics.
%		At the end of section \ref{tableaux}, we give a few examples for the tableau rules for some of these logics. In section \ref{upper} we give some general complexity results for the satisfiability problem of multi-agent justification logic and see how these results reflect to the example logics from section \ref{multidefinitions}. 
		%In section \ref{lower} we demonstrate that some of these examples are \PSPACE-hard and some are \EXP-hard.
%		Omitted proofs can be found in the appendix.
		%In section \ref{multidefinitions} we introduce logics $(J^{n}_{D,F,V,C})_\cs$ and we extend M-model semantics to this setting. In section \ref{sectionstar}, we define the $*$-calculus for $(J^{n}_{D,F,V,C})_\cs$ and in section \ref{algo}, we treat a special case, namely $D = \emptyset$, and prove an upper bound for the complexity of $(J^{n}_{\emptyset,F,V,C})_\cs$: $(J^{n}_{\emptyset,F,V,C})_\cs$ is in $\Pi_2^p$. This has been already established for all aforementioned justification logics(\cite{737242,Kuz08PhD}).\footnote{In this paper, logics are treated as the set of their theorems, so saying that a logic is in a complexity class $\C$ is the same as saying that the derivability problem for this logic is in $\C$.} The significance of this result is that, under some general assumptions, the complexity of a justification logic remains the same for one agent and for more than one agent. The corresponding question for the standard modal logic remains open.

		\section{Multiagent Justification Logic with Interactions}
		\label{multidefinitions}
		
%		We start by introducing the family of logics we are about to study. We provide the syntax of the formulas, the axioms and rules of each logic and their semantics. The proofs of the propositions that appear in this section are very similar to the proofs of corresponding results for the single-agent logics and the reader can see \cite{Art08RSL} or \cite{Kuz08PhD} for an overview. 
		%Therefore, they will be omitted, but they can be found in the appendix. 
		
		In this section we present the system we study in this paper, its semantics and the basic tools we will need later on. Most of the proofs for the claims here can be adjusted from the one- or two- agent versions of Justification Logic. The reader can see \cite{Art08RSL} or \cite{Kuz08PhD} for an overview of single-agent justification logic and \cite{Achilleos14TRtwoagent} for a two-agent version of this system.
		
		\subsection{Syntax}
		
		In this paper, if $n\in \nat$, $[n]$ will be the set $\{ 1, 2, \ldots , n \}$.
%%		
		%\begin{definition}
		For every $n \in \nat$, the justification terms of the language $L_n$ will include all constants $c_1, c_2, c_3, \ldots$ and variables $x_1, x_2, x_3, \ldots$ and if $t_1$ and $t_2$ are terms, then the following are also terms: $  [t_1 + t_2], [t_1\cdot t_2], ! t_1$. The set of terms will be referred to as $Tm$. We also use a set $SLet$ of propositional variables, or sentence letters. These will usually be $p_1,p_2,\ldots$.
		%\\	
		Formulas of the language $L_n$ include all propositional variables and if $\phi, \psi$ are formulas, $i \in [n] $ and $t$ is a term, then the following are also formulas of $L_n$: $\bot, \phi \rightarrow\psi , \term{t}{i} \phi $.  The remaining propositional connectives, whenever needed, are treated as constructed from $\rightarrow$ and $\bot$ in the usual way.
		% $\J_\infty  = \bigcup_{n \in \nat} L_n$
		%\end{definition}
%		
		The operators $\cdot, +$ and $!$ are explained by the following axioms. Intuitively, $\cdot$ applies a justification for a statement $A \rightarrow B$ to a justification for $A$ and gives a justification for $B$. Using $+$ we can combine two justifications and have a justification for anything that can be justified by any of the two initial terms - much like the concatenation of two proofs. 
		Finally, $!$ is a unary operator called the proof checker. Given a justification $t$ for $\phi$, $!t$ justifies the fact that $t$ is a justification for $\phi$.
		
%		\begin{definition}
			Let $n\in \nat$, $D,F\subseteq [n]$ and $V,C \subseteq [n]^{2}$. The logic $(J^{n}_{D,F,V,C})_{\emptyset}$ is the logic with modus ponens as a derivation rule and the following axioms:
			
			\begin{description}
				\item[Propositional Axioms:] Finitely many schemes of classical propositional logic;
				\item[Application:] $\term{s}{i}(\phi\rightarrow \psi) \rightarrow (\term{t}{i}\phi \rightarrow \term{[s\cdot t]}{i} \psi)$;
				\item[Concatenation:] 
				$\term{s}{i}\phi \rightarrow \term{[s + t]}{i} \phi$, 
				$\term{s}{i}\phi \rightarrow \term{[t + s]}{i} \phi$;
				\item[$F$-factivity:] For every $i\in F$, $\term{t}{i}\phi \rightarrow \phi$;
				\item[$D$-consistency:] For every $i\in D$, $\term{t}{i}\bot \rightarrow \bot$;
				\item[$V$-verification:] For every $(i,j)\in V$, $\term{t}{i}\phi \rightarrow  \term{! t}{j} \term{t}{i}\phi$;
				\item[$C$-conversion:] For every $(i,j)\in C$, $\term{t}{i}\phi \rightarrow  \term{t}{j}\phi$,
			\end{description}
			where in the above, $\phi$ and $\psi$ are formulas in $L_n$, $s, t$ are terms and $ i,j \in [n]$.
%		\end{definition}
		$F$-factivity and $D$-consistency are the usual factivity and consistency axioms for every agent in $F$ and $D$ respectively. Positive introspection is seen as a special case of $V$-verification - in this context, if agent $i$ has positive introspection, then $(i,i) \in V$.
		
			A constant specification for $J^{n}_{D,F,V,C}$ is any set \[ \mathcal{CS} \subseteq \{\term{c}{i} A \mid c \mbox{ is a constant, } A \mbox{ an axiom of } J^{n}_{D,F,V,C} \mbox{ and } i \in [n]\}. \]
			We say that axiom $A$ is justified by a constant $c$ for agent $i$, when $\term{c}{i}A \in \cs$.
			A constant specification is:
			\emph{axiomatically appropriate with respect to $I\subseteq [n]$} if for every $i \in I$, each axiom is justified by at least one constant,
			\emph{schematic} if every constant justifies only a certain number of schemes from the ones above (as a result, if $c$ justifies $A$ for $i$ and $B$ results from $A$ and substitution, then $c$ justifies $B$ for $i$) and
			\emph{schematically injective} if it is schematic and every constant justifies at most one scheme. 
			Let $cl_n(\cs)$ be the smallest set such that $\cs \subseteq cl_n(\cs)$ and for every $\term{t}{i}\phi \in cl_n(\cs)$, it is the case that for every $j\in [n]$, $\term{!t}{j}\term{t}{i}\phi \in cl_n(\cs)$.
			$(J_{D,F,V,C}^n)_\cs$ is $(J_{D,F,V,C}^n)_\emptyset + R4^{n}_\cs$, where $R4^{n}_\cs$ just outputs all elements of $cl_n(\cs)$.
%			\footnote{A motivation for $R4^{n}_\cs$ is that we want all agents to be aware of propositions provable in the logic.}

%		\begin{proposition}\label{consistency}
%			For any $n \in \nat$, $D, F \subseteq [n]$, $V, C \subseteq [n]^2$ and constant specification $\cs$, 
			
			$(J^n_{D,F,V,C})_\cs$ is consistent: just map each formula to the propositional formula that is the result of removing all terms from the original one; then, all axioms are mapped to propositional tautologies and modus ponens preserves this mapping.

		\subsection{Semantics}
		
		We now introduce models for our logic. In the single-agent cases, M-models (introduced in \cite{Mkr97LFCS,DBLP:conf/csl/Kuznets00}) and F-models (introduced in \cite{Fit05APAL,Pac05PLS,DBLP:conf/csr/Kuznets08}) are used (also in \cite{Achilleos14TRtwoagent} for two-agent logics) and they are both important in the study of complexity issues. In this paper we are mostly interested in F-models, which we will usually just call models. These are essentially Kripke models with an additional machinery to accommodate justification terms.
		
%		\begin{definition}
			Let $\J=(J^n_{D,F,V,C})_\cs$. Then, an F-model $\M$ for $\J$ is a quadruple $(W, (R_i)_{i \in [n]}, (\A_i)_{i \in [n]},\V)$, where $W \neq \emptyset$ is a set, for every $i\in [n]$, $R_i\subseteq W^2$ is a binary relation on $W$, $\V:SLet \To 2^{W}$ and for every $i\in [n]$, $\A_i:(Tm\times L_n) \To 2^{W}$. $W$ is called the \emph{universe} of $\M$ and its elements are the worlds or states of the model. $\V$ assigns a subset of $W$ to each propositional variable, $p$, and $\A_i$ assigns a subset of $W$ to each pair of a justification term and a formula. Furthermore, $(\A_{i})_{i\in[n]}$ will often be seen and referred to as $\A : [n]\times Tm \times L_n \To 2^{W}$ and $\A$ is called an admissible evidence function. Additionally, $\A$ must satisfy the following conditions:
			\begin{description}
				\item {Application closure:} for any $i\in [n]$, formulas $\phi, \psi$, and justification terms $t, s$, \\ 
%				\flushright	
				 \hfill 
				$ \A_i(s,\phi \rightarrow \psi) \cap \A_i(t,\phi) \subseteq \A_i(s\cdot t, \psi). $
				\item {Sum closure:} for any $i\in [n]$, formula $\phi$, and justification terms $t, s$, \\ 
				\hfill $ \A_i(t,\phi) \cup \A_i(s,\phi) \subseteq \A_i(t+s,\phi).$
				\item {$\cs$-closure:} for any 
				%correction:
						formula $\term{t}{i}\phi \in cl_n(\cs)$, $\A_i(t,\phi) = W$.
%				$i\in [n]$, $\many{k}{r+1} \in [n]$, axiom $A$, constant $c$, such that $\term{c}{i}A \in \cs$, \[ \A_j(\underbrace{!!\cdots !}_{r+1}c, \term{\underbrace{!\cdots !}_{r}c}{k_r} \ \cdots \ 
%				\term{!c}{k_1} \term{c}{i}A) = \A_i(c,A) = W. \]
				\item {$V$-Verification Closure:} If $(i,j)\in V$, then $\A_i(t,\phi) \subseteq \A_j(!t,\term{t}{i}\phi)$
				\item {$C$-Conversion Closure:} If $(i,j)\in C$, then $\A_i(t,\phi) \subseteq \A_j(t,\phi)$ 
				\item {$V$-Distribution:} for any formula $\phi$, justification term $t$, $(i,j)\in V$ and $a,b \in W$, if $a R_j b$ and $a \in \A_i(t,\phi)$, then $b \in \A_i(t,\phi)$.
			\end{description}
			The accessibility relations, $R_i$, must satisfy the following conditions:
			\begin{itemize}
				\item If $i \in F$, then $R_i$ must be reflexive.
				\item If $i \in D$, then $R_i$ must be serial ($\forall a \in W \ \exists b \in W \ a R_i b$).
				\item If $(i,j) \in V$, then for any $a,b,c\in W$, if $a R_j b R_i c$, we also have $a R_i c$.\footnote{Thus, if $i$ has positive introspection (i.e. $(i,i)\in V$), then $R_i$ is transitive.}
				\item For any $(i,j)\in C$, $R_j \subseteq R_i$.
			\end{itemize}
			Truth in the model is defined in the following way, given a state $a$:
			\begin{itemize}
				\item $M,a \not \models \bot$ and if $p$ is a propositional variable, then $\M,a \models p$ iff $a \in \V(p)$.
				\item If $\phi, \psi$ are formulas, then $\M,a \models \phi \rightarrow \psi$ if and only if $M,a \models \psi$, or $\M,a \not \models \phi$.
				\item If $\phi$ is a formula and $t$ a term, then $\M,a \models \term{t}{i}\phi$ if and only if $a \in \A_i(t,\phi)$ and $\M,b\models \phi$ for all $b \in W$ such that $a R_i b$.
			\end{itemize}
%		\end{definition}

	A formula $\phi \in L_n$ is called satisfiable if there are some $\M,a \models \phi$; we then say that $\M$ satisfies $\phi$ in $a$.		
		 If $\cs$ is axiomatically appropriate with respect to $D$, then $(J^n_{D,F,V,C})_\cs$ is sound and complete with respect to its models; it is also sound and complete with respect to its models that have the 
		 \emph{Strong Evidence Property: \emph{$\M,a \models \term{t}{i}\phi$ iff $a \in \A_I(t,\phi)$}}; furthermore, if $\phi$ is satisfiable, then it is satisfied by a model $\M$ of at most $2^{|\phi|}$ states - and in fact,  it is satisfied by a model $\M$ of at most $2^{|\phi|}$ states that has the strong evidence property (see \cite{Achilleos14TRtwoagent} for proofs of all the above that can be easily adjusted for this general case). A pair $(W,(R_i)_{i\in [n]})$ as above is called a frame for $(J^n_{D,F,V,C})_\cs$.

		\subsection{The $*$-calculus.}
%		\label{sec:starcalc}
%		
		We present the $*$-calculi for $(J^n_{D,F,V,C})_\cs$. The $*$-calculi for the single-agent justification logics have proven to be an invaluable tool in the study of the complexity of these logics. This concept and results were adapted to the two-agent setting in \cite{Achilleos14TRtwoagent} and here we extend them to the general multi-agent setting. Although the calculi have significant similarities to the ones of the single-agent justification logics, there are differences, notably that each calculus depends upon a frame and operates upon $*$-expressions (defined below) prefixed by states of the frame. 
		%For single-agent versions of $*$-calculi, the reader can see \cite{Mkr97LFCS,NKru06TCS,Kuz08PhD}. 
		A $*$-calculus was first introduced in \cite{NKru06TCS}, but its origins can be found in \cite{Mkr97LFCS}. 
%		The form on which the ones in this section are based is from \cite{Kuz08PhD}. Single-agent versions of the following results can be found there.

		If $t$ is a term, $\phi$ is a formula, and $i\in [n]$, then $*_i(t,\phi)$ is a star-expression ($*$-expression).
		Given a frame $\mathcal{F} = (W,(R_i)_{i\in [n]})$ and $V, C \subseteq [n]^2$ and constant specification $\cs$, the $*_{\cs}^{\mathcal{F}}(V,C)$ calculus on the frame $\mathcal{F}$ is a calculus on $*$-expressions prefixed by worlds from $W$ with the axioms and rules that are shown in figure \ref{fig:starcalc}.\\
		\noindent
		\begin{figure}
			\begin{tabular}{|c|c|}\hline
			\begin{minipage}{0.5\linewidth}
			\begin{description}
			\item[$*\mathcal{CS}(\mathcal{F})$ Axioms:]  $w \ *_i(t,\phi)$, where $\term{t}{i}\phi \in cl_n(\cs)$ and $w \in W$
			
			\vspace{2ex}
			%		and $$w \ *_i(\underbrace{!!\cdots !}_{r+1}c, \term{\underbrace{!\cdots !}_{r}c}{k_r} \ \cdots \ 
			%		\term{!c}{k_1} \term{c}{i}A),$$ where $\term{c}{i}A \in \mathcal{CS}$, for any $r \in \nat$ and $\manyr{k} \in [n]$, $w \in W$.
			\item[$*$App$(\F)$:]
			\[ 
			\inferrule*
			{w \ *_i(s,\phi\rightarrow \psi) \\ w \ *_i(t,\phi)}{w \ *_i(s\cdot t , \psi) }
			\] 
			\item[$*$Sum$(\mathcal{F})$:] \[ 
			\inferrule*
			{w \ *_i(t,\phi) }{w \  *_i(s+t,\phi) } \qquad
			\inferrule*
			{w \ *_i(s,\phi) }{w \  *_i(s+t,\phi) }
			\] 
			\end{description}
			\end{minipage}
			&
			\begin{minipage}{0.45\linewidth}
			\begin{description}	
			\vspace{2ex}
			\item[$*$V$(\mathcal{F})$:] For any $(i,j) \in V$, \[ 
			\inferrule*
			{w \ *_i(t,\phi) }{w \  *_j(! t,\term{t}{i}\phi) } 
			\] 
			\item[$*$C$(\mathcal{F})$:] For any $(i,j) \in C$, \[ 
			\inferrule*
			{w \ *_i(t,\phi) }{w \  *_j(t,\phi) } 
			\] 
			\item[$*$V-Dis$(\mathcal{F})$:] For any $(i,j) \in V$ and $(a,b)\in R_j$, \[ 
			\inferrule*
			{a \ *_i(t,\phi) }{b \  *_i(t,\phi) } 
			\] 	\vspace{1ex}
			\end{description}
			\end{minipage}
			\\ \hline
			\end{tabular}
			\caption{The $*^\F_\cs(V,C)$-calculus: where $\mathcal{F} = (W,(R_i)_{i\in [n]})$ and for every $i \in [n]$}\label{fig:starcalc}
		\end{figure}
		%\end{definition}
		%The concepts previously defined for M-models can be transferred to frames.
%		
		%\begin{definition}
		\indent
			Notice that the calculus rules correspond to the closure conditions of the admissible evidence functions. In fact and because of this, given a frame $\F = (W,R_1,R_2)$ and a set $S$ of $*$-expressions prefixed by states of the frame, the function $\A$ such that $\A_i(t,\phi) = \{w\in W | S\vdash_{*^\F_\cs(V,C)} w\ *_i(t,\phi)  \}$ is an admissible evidence function and in fact it is the minimal admissible evidence function such that for every $w\ *_i(t,\phi) \in S$, $w \in \A_i(t,\phi)$ in the sense that always $\A_i(t,\phi) \subseteq \A'_i(t,\phi)$ for any other admissible evidence function $\A'$ such that for every $w\ *_i(t,\phi) \in S$, $w \in \A'_i(t,\phi)$. 
%		A possible evidence function on frame $\mathcal{F} = (W,(R_i)_{i\in [n]})$ is any function $ \B : Tm\times L_n \times[n] \longrightarrow \pow(W). $
%		We say that a possible evidence function on frame $\mathcal{F} = (W,(R_i)_{i\in [n]})$, $\mathcal{B}_2$, is based on a possible evidence function on frame $\mathcal{F}$, $\mathcal{B}_1$, and write $\mathcal{B}_1 \subseteq \mathcal{B}_2$, if for all terms $t$, formulas $\phi$ and $i\in [n]$,  $\mathcal{B}_1(t,\phi,i) \subseteq \mathcal{B}_2(t,\phi,i). $
%		Let $\mathcal{EF}$ be a class of possible evidence functions on frame $\mathcal{F} = (W,(R_i)_{i\in [n]})$. A possible evidence function $\mathcal{B} \in \mathcal{EF}$ is called the minimal evidence function in $\mathcal{EF}$ if for all $\mathcal{B'} \in \mathcal{EF}$, $\mathcal{B} \subseteq \mathcal{B'}$.
%		For $\B$, a possible evidence function on frame $\mathcal{F} = (W,(R_i)_{i\in [n]})$,  $\B^* = \{w \ *_i(t,\phi)|w \in \B(t,\phi,i)\}$.
		% This definition can also be naturally extended to the case where $\B$ is a set of formulas of the form $\term{t}{i}\phi$.
		%\end{definition}
		Therefore, given a frame $\F = (W,R_1,R_2)$ and two set $\T,\F$ of $*$-expressions prefixed by states of the frame there is an admissible evidence function $\A$ on $\F$ such that  for every $w\ *_i(t,\phi) \in \T$, $w \in \A_i(t,\phi)$ and  for every $w\ *_i(t,\phi) \in \F$, $w \notin \A_i(t,\phi)$, if and only if there is no $f \in \F$ such that $\T \vdash_* f$.
		This observation  yields the following.
		\begin{proposition} %[Single-agent cases: \cite{NKru06TCS,Kuz08PhD}]
			\label{prp:proofbystarcalc}
			For any\footnote{Note that we actually need an axiomatically appropriate constant specification to have completeness with respect to F-models, so we cannot immediately conclude with this result for \emph{any} constant specification. Nevertheless, proposition \ref{prp:proofbystarcalc} holds, but to prove it we need to introduce M-models, but we do not use M-models anywhere else. Thus, the reader can see \cite{NKru06TCS,Kuz08PhD}, or \cite{Achilleos14TRtwoagent} for a proof of proposition \ref{prp:proofbystarcalc} for all constant specifications.} 
			constant specification $\cs$, frame $\F$ with universe $W$ and $w\in W$,  $ (J^n_{D,F,V,C})_\cs \vdash \term{t}{i}\phi \ \Longleftrightarrow \  \vdash_{*^\F_\cs(V,C)} w\ *_i(t,\phi)$.
		\end{proposition}
		
		%\begin{proof}[Proof of proposition \ref{prp:proofbystarcalc}]
		%	%	$\mathcal{AEF_\emptyset}$ is the set of all admissible evidence functions. 
		%	Let $w \in \A^m_i(t,\phi)  \Longleftrightarrow \vdash_{*_\cs(V,C)} w\ *_i(t,\phi).$ Also, let $\A^m_F = \{\term{t}{i}\phi \in L_n | \A^m_i(t,\phi) \neq \emptyset\}$. 
		%	Notice that the particular frame $\F$ and state $w$ do not matter, because we can easily see that for every other state $v$ and star expression $e$, $w\ e$ can be derived iff $v\ e$ can be derived (and thus, we can just ignore all parts of the derivation that mention the frame).
		%	By proposition \ref{thm:starminframes}, any admissible evidence function is based on $\A^m$. This means that all models satisfying the strong evidence property satisfy $\A^m_F$. Because of corollary \ref{strong}, all formulas in $\A^m_F$ are valid. Finally, if $\term{t}{i}\phi \not \in \A^m_F$, then this formula is not satisfied in any model with $\A^m$ as its admissible evidence function. 
		%\qed \end{proof}
		
		\begin{proposition}
			\label{thm:calccompnew}
			Let $\mathcal{CS}$ be a schematic constant specification in $\P$ and $V, C \subseteq [n]^2$. Then,
			the following problem is in $\NP$:
			%there exists a non-deterministic algorithm that runs in polynomial time and determines, 
			Given a finite frame $\mathcal{F} = (W, (R_i)_{i\in [n]})$ and a finite set $S$ of $*$-expressions prefixed by worlds from $W$, a formula $\phi$, a term $t$, a $w \in W$ and $i \in [n]$, is it the case that \[ S \vdash_{*^{\mathcal{F}}_{\cs}(V,C)} w \ *_i(t,\phi)\mbox{?} \]
		\end{proposition}
		The proof of this proposition is very similar to the one that can be found in \cite{Kuz08PhD}. What is different here is the additional assignment of a state set to each node of the derivation tree, which does not change things a lot.
		% and it is available in the appendix.
		
		\begin{proof} %[Proof of proposition \ref{thm:calccompnew}]
			%The algorithm for this version of the $*$-calculus will resemble the one presented for proposition \ref{thm:calccomp}. 
			For this proof and every $j \in [n]$, let $f_j:2^W\To 2^W$, s.t. for every $X \subseteq W$, $f_j(X) = \{y\in W | \exists x \in X, (j,j')\in V \ xR_{j'} y\}\cup X$. %\\
			%First, notice that the derivation tree for $w \ *_i(t,\phi)$ is bounded in the number of nodes by the number of subterms of $t$ multiplied by the number of worlds in $W$. So, the algorithm will be:
			\begin{itemize}
				\item Nondeterministically construct a rooted tree with pairs of the form $(j,s)$, where $j \in [n]$
				%	or appears in $t$ 
				and $s$ is a subterm of $t$, as nodes, such that $(i,t)$ is the root and the following conditions are met. Node $(j,s)$ can be the parent of $(j_1,s_1)$ or of both $(j_1,s_1)$ and $(j_2,s_2)$ as long as there is a rule $\frac{w\ *_{j_1}(s_1,\phi_1)}{w\ *_{j}(s,\phi_3)}$ or $\frac{w\ *_{j_1}(s_1,\phi_1) \ \ w\ *_{j_2}(s_2,\phi_2)}{w\ *_{j}(s,\phi_3)}$, respectively, of the $*$-calculus and as long as $(j_1,s_1) \neq (j,s) \neq (j_2,s_2)$. To keep this structure a tree, we can ensure at this step that there are no cycles, which would correspond to consecutive applications of $*$C$(\F)$ and which would be redundant.
				%		, where $j C j'$ and $j \neq j'$. Alternatively, $(j,s)$ can be the parent of  or that $j_1 = j_2$ and $s_2$ a proper subterm of $s_1$.
				%the nodes are compatible with the $*_\cs(V,C)$-calculus derivation rules
				\item Nondeterministically assign to each leaf, $(j,l)$, either 
				\begin{itemize}
				\item
					some formula $\psi$ and the closure under $f_j$ of some set $W'\subseteq W$, s.t. for every $w \in W'$, $w\ *_j(l,\psi) \in S$ or, 
				\item
					as long as $l$ is of the form $\underbrace{!\cdots !}_{k} c$, where $c$ a constant, $k \in \nat$, then we can also assign some $ \term{\underbrace{!\cdots !}_{k-1} c}{i_{k}} \cdots \term{!c}{i_2} \term{c}{i_1}  A$ and $W' = W$, where $A$ an axiom scheme, s.t. $\term{c}{i_1}A \in \cs$.
				\end{itemize}
				\item If for some node $\nu = (j,s)$ all its children, say $\nu_1 = (j_1,s_1), \nu_2 = (j_2,s_2)$ have been assigned some scheme or formula $P_1, P_2$ and world sets $V_1, V_2$, assign to $\nu$ some scheme or formula $P$, such that $P_1,P_2$ can be unified to $P'_1,P'_2$ such that $\frac{w *_{j_1}(s_1,P'_1) \ w *_{j_2}(s_2,P'_2)}{w *_{j}(s,P)}$ is a rule in the $*_\cs^{\F}(V,C)$-calculus and world set $V$, where $V$ is the closure of $V_1 \cap V_2$ under $f_j$. 
				%unify these in order to make the maximal subtree with $\nu$ as its root a valid  $*_\cs(V,C)$-calculus derivation. 
				Apply this step until the root of the tree has been assigned some scheme or formula and a $W'$ subset of $W$.
				\item Unify $\phi$ with the formula assigned to $(i,t)$ and verify that $w\in W'$ 
			\end{itemize}
		
			If some step is impossible, the algorithm rejects. Otherwise, it accepts.
			Using efficient representations of schemes using DAGs and Robinson's unification algorithm, the algorithm runs in polynomial time. We can see that as the tree is constructed, if $(i,s)$ is assigned scheme $P$ and set $V$, then the construction effectively describes a valid derivation of any expression of the form $v\ *_i(s,\psi)$, where $v \in V$ and $\psi$ an instance of $P$.
			Therefore, if the algorithm accepts, there exists a valid $*^{\mathcal{F}}_{\cs}(V,C)$-calculus derivation of $w \ *_i(t,\phi)$. On the other hand if there is some $*^{\mathcal{F}}_{\cs}(V,C)$-calculus derivation for $w \ *_{i}(t,\phi)$ from $S$, then the algorithm in the first two steps can essentially describe this derivation by producing the derivation tree and the formulas/schemes by which the derivation starts. Therefore, the algorithm accepts if and only if there is a $*^{\mathcal{F}}_{\cs}(V,C)$-calculus derivation for $w \ *_{i}(t,\phi)$ from $S$. See \cite{NKru06TCS} and \cite{Kuz08PhD} for a more detailed analysis.
		\qed \end{proof}
		
%		\begin{note} \label{obs:fptstarcalc}
			The number of nondeterministic choices made by the algorithm in the proof of proposition \ref{thm:calccompnew} is bounded by $|t| + |S'|$, where $S' = \{*_j(s,\psi) | \exists w\ *_j(s,\psi) \in S  \}$. Therefore, if there is some formula $\psi$ such that $\term{t}{i}\phi$ is a subformula of $\psi$ and  for every $*_j(s,\psi') \in S'$, $\term{s}{j}\psi'$ is a subformula of $\psi$, then $2|\psi| \geq |t| + |S'|$ and therefore we can simulate all nondeterministic choices in time $2^{O(|\psi|)}$. Thus the algorithm can be turned into a deterministic one running in time $2^{O(|\phi|)}\cdot O(|W|^2)$. 
			This observation, the fact that a satisfiable $\phi$ can be satisfied by a model of at most $2^{|\phi|}$ states (see the previous subsection) and the previous two propositions give the following results:
%		\end{note}
%		
		\begin{corollary} Let $\J = (J^n_{D,F,V,C})_\cs$, where $\cs \in \P$ is schematic. Then,
			\begin{enumerate}
				\item Deciding for $\term{t}{i}\phi$ that $\J \vdash \term{t}{i}\phi$ is in \NP.
				\item If $\cs$ is axiomatically appropriate with respect to $D$, then the satisfiability problem for $\J$ is in \NEXP.
			\end{enumerate}
		\end{corollary}
 Additionally notice that if the term $t$ has no $+$, $\cs$ is schematically injective and $S = \emptyset$, we have essentially eliminated nondeterministic choices from the procedure above. Thus, we conclude (for the original result, see \cite{DBLP:conf/tark/ArtemovK09}):
		\begin{corollary} 
			Let $\J = (J^n_{D,F,V,C})_\cs$, where $\cs \in \P$ is schematically injective. Then,
			 deciding for $\term{t}{i}\phi$, where $t$ has no `$+$', that $\J \vdash \term{t}{i}\phi$ is in \P.
		\end{corollary}

	\section{Tableaux}\label{tableaux}
		In this section we give a general tableau procedure for every logic which varies according to each logic's parameters. We can then use the tableau for a particular logic and make observations on its complexity, as we do in the following section. A version of the tableau which is more efficient for some cases follows after. To develop the tableau procedure we need to examine the relations on the agents more carefully than we have so far. For this section and the following one fix some $J = (J^n_{D,F,V,C})_\cs$ and  we assume $\cs$ is axiomatically appropriate with respect to $D$ and schematic.
%		\subsection{Interactions Between Agents}
%		In this subsection we take a closer look at the interactions between the agents. This will give us the appropriate tools we need in the following. 

\subsection{A Closer Look on the Agents and their Interactions}
		 If $\manyk{A}$ are binary relations on the same set, then $A_1\cdots A_k$ is the binary relation on the same set, such that $x A_1 \cdots A_k y$ if and only if there are $\manyk{x}$ in the set, such that $x = x_1 A_1 x_2 A_2 \cdots A_{k-1} x_k A_k y$.  If $A$ is a binary relation, then $A^*$ is the reflexive, transitive closure of $A$; if $A$ is a set (but not a set of pairs), then $A^*$ is the set of strings from $A$.	We also use the following relation on strings: $a \sqsubseteq b$ iff there is some string $c$ such that $ac = b$.
%		
		%\begin{definition}
		We  define the following subsets of and relations on $[n]$.
		\begin{description}
%			\item[$C^*$] is the reflexive, transitive closure of $C$, and in general, if $A$ is a binary relation, then $A^*$ is the reflexive, transitive closure of $A$; if $A$ is a set (but not a set of pairs), then $A^*$ is the set of strings from $A$;
			\item[$ S = S(J) =$] $ \{ i \in [n] | (\exists j \in D \cup F) \ i C^* j \}$ attempts to capture exactly those agents that require a serial accessibility relation; from now and on these agents monopolize our attention;
			\item[$R = R(J) =$] $ \{ i \in [n] | (\exists j \in F) \ i C^* j \}$ attempts to capture exactly those agents that require a reflexive accessibility relation;
			\item[$C_F = $] $C\cup \{(i,j)\in V| i \in R, j\in S\}$ notice that if $iC_F j $, and $x R_j y$, then $xR_i y$;\footnote{The $F$ in $C_F$ is to indicate that $C_F$ is a variation of $C$ influenced by the agents in $F$, not that it is available with other subscripts}
			%\item[] $(C|_S)^{*}$ is the reflexive, transitive closure of $C|_D$, the restriction of $C$ to $S$.
			\item[$Q = $] $(V|_S\cup C|_S)^*$:  if $i Q j$, then $i$'s justifications somehow affect $j$'s justifications; conversely, $j$'s accesibility relations somehow affect $i$'s accessibility relations;
			\item[$\equiv_C = $] $C^*_F \cap (C_F^*)^{-1}$: $i \equiv_C j$ if and only if $i C_F^* j$ and $j C_F^* i$; we can easily see that $\equiv_C$ is an equivalence relation and that if $i \equiv_{C} j$, then $i$ and $j$ have the same accessibility relations; similarly,
%			\item[$\equiv_{VC} = Q \cap Q^{-1}$.]
		\end{description}
		
		For the above equivalence relations we can define equivalence classes on $S$, $P_C = \{\many{L}{k_C}\}$. 
%		$P_{VC} = \{\many{P}{k_{VC}}\}$.
%		Notice that as $C^* \subseteq (V\cup C)^*$, $\equiv_C \subseteq \equiv_{VC}$ and therefore $P_C$ is a refinement of $P_{VC}$. We'll call $P = P_{VC}$. 
%		Furthermore, notice that for any $L \in P$, either $\exists x,y \in L$ s.t. $x V y$, or $L \in P_C$. In the first case, $L$ will be called a \emph{V-class} of agents and in the second case it will be called a \emph{C-class} of agents. For each agent $i \in [n]$, $P(i)$ will be the equivalence class $L \in P$ s.t. $i \in L$ and 
		$\chi(i)$ is the equivalence class $L \in P_C$ s.t. $i \in L$. 
%		To help keep the relationship between these sets in mind, notice that $i \in \chi(i) \subseteq P(i) \subseteq S \subseteq [n]$.
		
		We can define relations $<_C$ and $<_{VC}$ on $P_C$ in the following way: $P_1 \leq_C P_2$ iff $\exists x \in P_1 \exists y \in P_2$ s.t. $x C_F^* y$ and $P_1 \leq_{VC} P_2$ iff $\exists x \in P_1 \exists y \in P_2$ s.t. $x Q y$. Also,  $P_1 <_{C} P_2$ iff  $P_1 \leq_{C} P_2$ and  $P_1 \not\leq_{C} P_2$ 
		and similarly for  $P_1 <_{VC} P_2$. Then, define $P_1 \leq_V P_2$ iff there are $x \in P_1$ and $y \in P_2$ s.t. $x Q V Q y$, that is, there are $ x_1, x_2 \in S$, where $x Q x_1 V x_2 Q y$. $P_1 <_V P_2$ iff $P_1 \leq_V P_2$ and $P_2 \not \leq_V P_1$.

%				We can define relations $<_C$ and $<_{VC}$ on $P$ and on $P_C$ in the following way: $P_1 \leq_C P_2$ iff $\exists x \in P_1 \exists y \in P_2$ s.t. $x C_F^* y$ and $P_1 \leq_{VC} P_2$ iff $\exists x \in P_1 \exists y \in P_2$ s.t. $x Q y$. Also,  $P_1 <_{C} P_2$ iff  $P_1 \leq_{C} P_2$ and  $P_1 \not\leq_{C} P_2$ 
%				and similarly for  $P_1 <_{VC} P_2$. Then, define $P_1 \leq_V P_2$ iff there are $x \in P_1$ and $y \in P_2$ s.t. $x Q V Q y$, that is, there are $ x_1, x_2 \in S$, where $x Q x_1 V x_2 Q y$. $P_1 <_V P_2$ iff $P_1 \leq_V P_2$ and $P_2 \not \leq_V P_1$. 
		
		%The following definition is needed in this section. % \ref{tableaux}.
		%
		%\begin{definition}
		Let $i \in [n]$. Then, $M_C(i) = M_C(\chi(i)) = \{L\in P_C| \chi(i) \leq_C L$ and $ \not\exists L'\in P_C $ s.t. $ L <_{C} L' \}$.
%		 and $M(i) = \{j \in S| P(j) $ is $\leq_{VC}$-maximal s.t. $ P(i) \leq_{VC} P(j) $, $j \notin R$ and $ P_C(j) \in M_C \}$.
		%\end{definition}
		%\begin{definition}
	
		%\end{definition}

		%\begin{definition}
%		We recursively define relation $\rightarrow$ on $S^*$: 
%		\begin{itemize}
%			\item if $jC_F i$ then $i \rightarrow j$;
%			\item if  $j V i$, then $ij \rightarrow j$;
%			\item If if $\beta \rightarrow \delta$, then $\alpha \beta \gamma \rightarrow \alpha \delta \gamma$.
%		\end{itemize} 
%		$\rightarrow^*$ is the reflexive, transitive closure of $\rightarrow$.
%		%\end{definition}
%		$\rightarrow^*$ tries to capture the closure of the conditions on the accessibility relations of a frame. This is made explicit with the following observation.
%		
%		\begin{note}
%			If in some frame $a R_{i_1}R_{i_2}\cdots R_{i_k}b$ and $i_1i_2\cdots i_k \rightarrow^* j_1j_2\cdots j_l$, then $a R_{j_1}R_{j_2}\cdots R_{j_l}b$. Furthermore, if, in addition, $l = k$, then for every $r \in [k]$, $j_r C_F^* i_r$.
%		\end{note}

		\subsection{The Tableau Procedure.}
		
		%We define the following: \[ N = \{L \in P | L \mbox{ is a C-class}\} \cup \]\[\{L[l]|L \in P, l \in P_C \ l \subseteq L, L \mbox{ is not a C-class, } l \mbox{ is a } \leq_{C}\mbox{-maximal subclass of } L\}.\]
		%Two sets of tableaux rules are presented and facts about them are proven in this section. 
		The formulas used in the tableau will have the  form 
		%$ \emptyset.\sigma\ T\ \Diamond_i$ or 
		$ \emptyset.\sigma \ s \ \beta \psi$, where $ \psi \in L_n$ or is a $*$-expression, $\sigma \in P_C^*$ (the world prefixes are strings of equivalence classes of agents), $\beta$ is (either the empty string or) of the form $\Box_i \Box_j \cdots \Box_k$, $i,j,\ldots,k \in [n]$, and $s \in \{T,F\}$. Furthermore, $\emptyset.\sigma$ will be called a world-prefix or state-prefix, $s$ a truth-prefix and world prefixes will be denoted as $\emptyset.s_1.s_2\ldots s_k$, instead of $\emptyset.s_1s_2\cdots s_k$, where for all $x \in [k]$, $s_x \in P_C$.
		%to be $P(i)$ if $P(i)$ is a $C$-class and $P(i)[P_C(i)]$, if $P(i)[P_C(i)] \in N$, if not. Notice that sometimes there is no $\chi(i)$.
		
		A tableau branch is a set of formulas of the form $ \sigma \ s \ \beta \psi$, as above. A branch is complete if it is closed under the tableau rules (they follow). It is propositionally closed if $ \sigma \ T \ \beta \psi$ and $ \sigma \ F \ \beta \psi$ are both in the branch. We say that a tableau branch is constructed by the tableau rules from $\phi$, if it is a closure of $\{\emptyset\ T\ \phi \}$ under the rules.
		
		%\begin{definition}
		%	
		%\end{definition}
		
		%\begin{definition}
%		For every agent $i \in S$, we introduce a new agent, $\overline{i}$ and we extend $\rightarrow^*$, so that  $\overline{i} \rightarrow^* \chi$ for every $\chi \in (P(i))^*$ such that $\chi \rightarrow^* i$.
%		%	and for every $j,k \in S$, where $k \equiv_{VC} i$, $jk\rightarrow^* k$ iff $j\overline{i} \rightarrow^* \overline{i}$. 
%		For each $L \in P_C$, we fix some $i \in P_C$ and $\overline{L} = \overline{i}$. Furthermore, if $xy \in P_C^* \cup S^*$, then $\overline{xy} = \overline{x}$ $\overline{y}$.
%		%\end{definition}
%		
%		
%		%\begin{definition}
%		Let 
%		%$b$ be a branch, 
%		$L \in P$ and $\sigma$ be a finite string of elements from  $P_C$. Then, $L$ is \emph{visible} from $\emptyset.\sigma$ if and only if there is some $\chi(i)\subseteq L$, some $\chi \in P_C^*$ and some $\alpha \in S^*$ such that $\sigma = \tau.\chi(i).\chi$ and $\overline{\chi} \alpha \rightarrow^* i$.
		%$\emptyset.\sigma \ T\ L' \in b$. 
		%is of the form $\lambda.L'.\chi$, where if $\chi(j)$ appears in $\chi$, then $L \leq_{VC} P(j)$. In this case, the \emph{$L$-view} from $\sigma$ is $\lambda.L'$. 
		%\end{definition}
		
		%\begin{obs}
		%	If $\chi = \chi(x_1).\chi(x_2)\ldots \chi(x_k)$ and $\overline{\chi} \rightarrow^* \beta \in S^*$, then there is some $S^* \ni \alpha \rightarrow^* x_1x_2\cdots x_k$ and $\alpha \rightarrow^* \beta$.
		%\end{obs}
		
		%\begin{definition}
		For every $i,j \in S$, $i \in N(j)$ if
		there are some $i', i'' \in S\smallsetminus R$  such that $i' \equiv_{VC} i''$, $i'' VC_F^* j$ and $\chi(i) \in M_C(i')$.
		%\end{definition}
		
		%\begin{definition}
		%	We define relation $N \subseteq S^2$. $j N i$ exactly when the following is true: 
		%	\begin{itemize}
		%	\item if $P(j)$ a V-class, then $j = \min (P(j))_V$, and 
		%	\item there are some $k V j' C_F^* j$ and 
		%	\item if $P(k)$ not a V-class, then  $\chi(i) = f_m(\chi(k))$, and 
		%	\item if $P(k)$ a V-class, then  either 
		%	\begin{itemize}
		%		\item
		%	$i \in (P(k))_V$, or 
		%	\item $P(k) = P(j)$ and there is some $k' \in (P(k))_G$ such that $f_m(\chi(k')) = \chi(i)$.
		%\end{itemize}
		%\end{itemize}
		%\end{definition}
		
		The tableau rules will include certain classical rules to cover propositional cases of formulas, as well as the ones that follow:
		
		\noindent 
		\begin{minipage}[l][1.7cm]{0.4\linewidth}
			\[ 
			\inferrule*[right=TrB]{\sigma\ T\ \term{t}{i}\psi}{\sigma\ T\ *_i(t,\psi) \\\\ \sigma\ T\ \Box_i\psi }  
			\]
		\end{minipage}
		\hfill
		\begin{minipage}{0.55\linewidth}
			if $i \in S$;
		\end{minipage} 
		\begin{minipage}[l][1.3cm]{0.4\linewidth}
			\[ 
			\inferrule*[right=TrD]{\sigma\ T\ \term{t}{i}\psi}{  \sigma.\chi(j)\ F\ \bot }  
			\]
		\end{minipage}\hfill
		\begin{minipage}{0.55\linewidth}
			if $i \in S$, $\chi(j) \in M_C(\chi(i))$, and $j \notin R$;
		\end{minipage}
		\begin{minipage}[l][1.3cm]{0.4\linewidth}
			\[ \inferrule*[right=Tr]{\sigma\ T\ \term{t}{i}\psi}{\sigma\ T\ *_i(t,\psi)}
			\]
		\end{minipage}\hfill
		\begin{minipage}{0.55\linewidth}
			if $i \not\in S$;
		\end{minipage}
		\begin{minipage}[l][1.3cm]{0.4\linewidth}
			\[ 
			\inferrule*[right=Fa]{\sigma\ F\ \term{t}{i}\psi}{\sigma\ F\ *_i(t,\psi)} 
			\]
		\end{minipage}\hfill
		\begin{minipage}{0.55\linewidth}  \phantom{I like cheese}
		\end{minipage}
		\begin{minipage}[l][1.3cm]{0.4\linewidth}
			\[
			\inferrule*[right=S]{\sigma.\chi(j) \ F \ \bot }{\sigma.\chi(j).\chi(i) \ F \ \bot} 
			\]
		\end{minipage}\hfill
		\begin{minipage}{0.55\linewidth}
			if $i \in N(j)$;
%			and it is not the case that 
%			$P(i)$ is a $V$-class visible from $\sigma.\chi(j)$;
		\end{minipage}
		\begin{minipage}[l][1.3cm]{0.4\linewidth}
			\[
			\inferrule*[right=SB]{\sigma\ T\ \Box_{i}\psi}{\sigma.\chi(i)\ T\ \psi} 
			\] 
		\end{minipage}\hfill
		\begin{minipage}{0.55\linewidth}
			if $\sigma.\chi(i)$ has already appeared;
		\end{minipage}
		\begin{minipage}[l][1.3cm]{0.4\linewidth}
			\[ 
			\inferrule*[right=FB]{\sigma\ T\ \Box_{i}\psi}{\sigma\ T\ \psi} 
			\] 
		\end{minipage}\hfill
		\begin{minipage}{0.55\linewidth}
			if $i \in F$;
		\end{minipage}
		\begin{minipage}[l][1.3cm]{0.4\linewidth}
			\[ 
			\inferrule*[right=C]{\sigma\ T\ \Box_{i}\psi}{\sigma\ T\ \Box_{j}\psi}  
			\]
		\end{minipage}\hfill
		\begin{minipage}{0.55\linewidth}
			if $i C j$;
		\end{minipage}
		\begin{minipage}[l][1.3cm]{0.4\linewidth}
			\[
			\inferrule*[right=V]{\sigma\ T\ \Box_{i}\psi}{\sigma\ T\ \Box_{j}\Box_i\psi} 
			\] 
		\end{minipage}\hfill
		\begin{minipage}{0.55\linewidth}
			if $i V j$.
		\end{minipage}

		We do not explicitly mention it anywhere else, but of course, we need a set of rules to cover propositional cases as well. In particular we can use\\
		\begin{minipage}[l][1.5cm]{0.3\linewidth}
			\[ 
			\inferrule*{\sigma\ T\ \psi\rightarrow \psi'}{\sigma\ F\ \psi\quad \mid\quad \sigma\ T\ \psi'}  
			\]
		\end{minipage}\hfill and \hfill
		\begin{minipage}[l][1.5cm]{0.3\linewidth}
			\[
			\inferrule*{\sigma\ F\ \psi\rightarrow \psi'}{\sigma\ T\ \psi \\\\ \sigma\ F\ \psi'}  
			\] 
		\end{minipage}
		
		The separator $|$ indicates a nondeterministic choice between the two options it separates.
		
		%\begin{note} Notice that $\sigma.\chi(i)\ F\ \bot $ can only occur from rules S and TrD. TrD requires that $i \in M(i')$ for some $i'$, but by definition, $M(i')\cap R = \emptyset$. Rule S requires that $i \in N(j)$ for some $j$, which in turn by definition requires that there is some $i' \notin R$ such that $\chi(i) \in M_C(\chi(i'))$. Since if $i \in R$ then we can conclude that $i' \in R$, we again have that $i \notin R$. We can thus conclude that if $i \in R$, then no formula of the form $\sigma.\chi(i)\ \alpha$ is produced. \label{obs:nor}
		%\end{note}
		
%		\begin{definition}\label{dfn:frameconstruction}
			If $b$ is a tableau branch, then\footnote{$\emptyset.P_C^*$ here and wherever else it may appear is the set $\{\emptyset.x|x\in P_C^*\}$.} $W(b) = \{\sigma \in \emptyset.P_C^* | \mbox{ there is some } \sigma\ a \in b \}.$ 
			Let $(R_i)_{i \in [n]}$ be such that for every $i \in [n]$,
			\[R_i = \{(\sigma,\sigma.\chi(i))\in (W(b))^2\} \cup  \{ (w,w)\in (W(b))^2 | i \in F \} \]
%			\[ \cup \]
%			\[\{(\sigma,\tau.\chi(i))\in (W(b))^2|
%			P(i) \mbox{ a }V\mbox{-class, $\tau.\chi(j)$ the $P(i)$-view from $\sigma$ for some } j\},\]
%			%	\[\cup\]
			%	\[\{(\sigma.\chi(l),\sigma.\chi(j).\chi(i))\in (W(b))^2|	P(l) = P(j) \mbox{ a }V\mbox{-class and } j N i \}\]
			%	\[ iC_F^* i' V C_F^* l,\ \sigma.\chi(j).\chi(i) [\chi(i')] \in b\}\]
			%	\[\cup \]
			%	\[ \{ (w,w)\in (W(b))^2 | i \in F \}, \] 
			then $\F(b) = (W(b),(R'_i)_{i\in [n]})$, where $(R'_i)_{i\in [n]}$ is the closure of $(R_i)_{i\in [n]}$ under the conditions of frames for the accessibility relations, except for seriality. $(R'_i)_{i\in [n]}$ is constructed in the following way: for every $i \in [n]$, let $R^0_i = R_i$ and for every $k \in \nat \cup \{0\}$, \[R_i^{k+1} = R_i^k \ \cup \ \bigcup_{(i,j) \in C} R_j^k \ \cup \ \bigcup_{(i,j) \in V}\{(a,b) \in (W(b))^2 | \exists (a,c) \in R_j^k, (c,b) \in R_i^k \} \] and then, $\F(b) = (W(b),(\bigcup_{k \in \nat}R_i^k)_{i \in [n]})$.
			
			Finally, let $T(b) = \{ \sigma \ *_i(t,\psi) | \sigma\ T \ *_i(t,\psi) \mbox{ appears in }b\}$ and $F(b) = \{ \sigma \ *_i(t,\psi) | \sigma\ F \ *_i(t,\psi) \mbox{ appears in }b\}$. 
%		\end{definition}
		
		A branch $b$ of the tableau is rejecting when it is propositionally closed or there is some $f \in F(b)$ such that $T(b) \vdash_{*_\cs^{\F(b)}(V,C)} f$. Otherwise it is an accepting branch.

%		\begin{note} \label{lem:vedges}
			By induction on the construction of $\F(b)$, it is not hard to see that 
%			given a tableau branch $b$, if $\F(b) = (W,(R_i)_{i \in [n]})$, then 
			for every $(\sigma,\tau.\chi(j))\in R_i$, it must be the case that $i C_F^* j$ or that $i \in R$ and $\sigma = \tau.\chi(j)$.
%		\end{note}
%		
%		
		%\begin{lemma}\label{lem:fewvclasses}
		%	Let $i, j \in L \in P$, where $L$ a $V$-class. If there is some $\sigma \in W$ of the form $\tau.\chi(i).\upsilon.\chi(j)$, then $L$ is not visible from $\tau.\chi(i).\upsilon$.
		%\end{lemma}
		%\begin{proof}
		%	%Assume for simplicity that $L$ is a $C$-class, as this does not affect the proof.
		%	The only way to introduce prefix $\tau.\chi(i).\upsilon.\chi(j)$ is 
		%%	by a formula of the form $\tau.\chi(i).\upsilon\ T\ \Diamond_k$, where $k \in L$, this formula must have been introduced 
		%	through rule $TrD$ or $S$. But since $L$ is visible from $\tau.\chi(i).\upsilon$, this cannot happen. 
		%\qed \end{proof}
%		
%		\begin{lemma}\label{lem:boxtoedge}
%			Let $b$ be a complete tableau branch and $\F(b) = (W, (R_i)_{i\in [n]})$. 
			By induction on the frame construction we can see that
			if $\sigma \ T \ \Box_i \phi$ appears in $b$ and $\sigma R_i \tau$, then $\tau \ T \ \phi$ appears in $b$.
%		\end{lemma}
		
		%	\begin{obs}\label{obs:backwardsboxes}
		%		We can also claim the following: Let $b$ be a complete tableau branch constructed by the tableau rules from $\phi$ and $\F(b) = (W, (R_i)_{i\in [n]})$. If $\sigma \ T \  \Box_{i_1}\cdots \Box_{i_k}\psi$ appears in $b$, then there is some $\tau \in W$, term $t$ and $i \in S$ such that $\tau R_i \upsilon$ for every $\upsilon \in W$ such that $\sigma R_{i_1}\cdots R_{i_k} \upsilon$ and $\tau \ T \ \term{t}{i}\phi$ appears in $b$. We can prove this by induction on the tableau derivation and a simpler consequence is that: If $\sigma \ T \ \psi$ appears in $b$ and $\phi \neq \psi$, then there is some $\tau \in W$, term $t$ and $i \in S$ such that $\tau R_i \sigma$ and $\tau \ T \ \term{t}{i}\phi$ appears in $b$.
		%	\end{obs}
		%	
			\begin{proposition}\label{prp:tableaucomplete}
				If there is a complete accepting tableau branch $b \ni \emptyset\ T\ \phi$, then the formula $\phi$ is satisfiable by a model for $J$.
			\end{proposition}
			
			\begin{proof}
				Let $\M = (W,(R_i)_{i\in [n]}, \A, \V)$, where $(W,(R_i)_{i\in [n]}) = \F(b)$, $\V(p) = \{w \in W| w\ T\ p \in b \}$, and $\A_i(t,\psi) = \{w \in W| T(b) \vdash_{*^\F_\cs(V,C)} w\ *_i(t,\psi) \}$.
		
				Let $\M' = (W,(R'_i)_{i\in[n]},(\A_i)_{i\in[n]},\V)$, where for every $i \in [n]$, if $i \in S$, then $R'_i = R_i \cup \{(a,a)\in W^2 |\exists j \in S $ s.t. $ iC_F^* j, \not\exists (a,b) \in R_j \}$ and $R'_i = R_i$, otherwise. $\M'$ is an F-model for $J$: $(\A_i)_{i\in[n]}$ easily satisfy the appropriate conditions, as the extra pairs of the accessibility relations do not affect the $*$-calculus derivation, and we can prove the same for $(R'_i)_{i\in[n]}$. If $aR'_ibR'_jc$ and $j V i$, if $(a,b) \in R'_i\smallsetminus R_i$, then $a = b$ and thus $a R'_i c$. If $(a,b) \in R_i$, then, from rule S, there must be some $(b,c')\in R_j$, so $(b,c)\in R_j$ and thus, $(a,c) \in R_j$. If $(a,b) \in R'_i$ and $j C i$, then, trivially, whether $(a,b) \in R_i$ or not, $(a,b) \in R'_j$. 
				%	If, on the other hand, $(a,b) \in R'_i\smallsetminus R_i$, then it must also be the case that $(a,b) \in R'_j\smallsetminus R_j$, because of rules SD and S. 
				
				By induction on $\chi$, we prove that for every formula $\chi$ and $a \in W$, if $a\ T\ \chi \in b$, then $\M',a\models \chi$ and if $a\ F\ \chi \in b$, then  $\M',a \not\models \chi$. Propositional cases are easy. If $\chi = \term{t}{i}\psi$ and $a\ F\ \chi \in b$, then $a \notin \A_i(t,\psi)$, so $\M', a \not \models \chi$. On the other hand, if $a\ T\ \term{t}{i}\psi \in b$, then $a \in \A_i(t,\psi)$ and by rule TrD,
		%		(and because $C_F^*(VC_F^*)^*  = Q$),
				 for every $j \in S$ such that $i C_F^* j$, there is some $(a,b)\in R_j$. Therefore, for every $(a,b) \in R_j'$, it is the case that $(a,b) \in R_j$, so by rule TrB, a previous observation about formulas of the form $w\ T\ \Box_i \alpha$, and the inductive hypothesis, for every $(a,b)\in R_i$, $\M',b \models \psi$ and therefore, $\M',a \models \term{t}{i}\psi$. 
			\qed \end{proof}
		Now we can prove the following proposition.
		\begin{proposition}\label{prp:tableaux}
			Let $\phi \in L_n$. $\phi$ is $(J^n_{D,F,V,C})_\cs$-satisfiable if and only if there is a complete tableau branch $b$ that is produced from $\emptyset \ T \ \phi$, such that 
			\begin{itemize}
				\item
				for all $\sigma, \alpha$, not both $\sigma\ T\ \alpha$ and $\sigma \ F \ \alpha$ appear in $b$ and 
				\item
				For any $\beta \in F(b)$, $T(b) \not \vdash_{*^\F_\cs(V,C)} \beta$.
			\end{itemize}
		\end{proposition}
		
			\begin{proof} %[Proof of proposition \ref{prp:tableaux}]
				The ``if'' direction was handled by proposition \ref{prp:tableaucomplete}. We will prove the ``only if'' in the following.
				Let $\M$ be an F-model, $ (W, (R_i)_{i\in [n]},(\A_i)_{i\in [n]},\V)$ 
%				of at most $2^{|\phi|}$ states 
				that has the strong evidence property and a state $s \in W$ such that $\M, s \models \phi$. 
				Furthermore, 
%				for every $a \in W$ and V-class $L$ fix some $L$-cluster for $a$ and also 
				fix some $\cdot^\M : \emptyset.P_C^* \To W$, such that $\emptyset^\M = s$ and the following conditions are met.  For any $\sigma.\chi(i)\in P_C^*$, $(\emptyset.\sigma.\chi(i))^\M$ is some element of $W$ s.t. $((\emptyset.\sigma)^\M,(\emptyset.\sigma.\chi(i))^\M) \in R_i$.
%				while if $P(i)$ is a $V$-class, then for  $(a_j)_{j\in L}$, the fixed $L$-cluster for $\sigma^\M$,  $(\sigma.\chi(i))^\M = a_i$.
				
				Let $L_n^\Box = \{\Box_{i_1}\cdots \Box_{i_k} \phi | \phi \in L_n, k \in \nat, \manyk{i} \in [n] \}$. Given a state $a$ of the model,  and $\Box_i\psi \in sub_\Box(\phi)$, $\M, a \models \Box_i\psi$ 
								has the usual, modal interpretation, 
								$\M, a \models \Box_i\psi$ iff 
								for every $(a,b) \in R_i$, $\M, b \models \psi$
				
				We can see in a straightforward way and by induction on the tableau derivation that there is a branch, such that if $\sigma \ T \ \psi$ appears in the branch and $\psi \in L_n^\Box$, then $\M,\sigma^\M \models \psi$, if $\sigma \ F \ \psi$ appears in the branch and $\psi \in L_n^\Box$, then $\M,\sigma^\M \not\models \psi$, if $\sigma \ T \ *_i(t,\psi)$ appears in the branch, then $\sigma^\M \in \A_i(t,\psi)$ and if $\sigma \ F \ *_i(t,\psi)$ appears in the branch, then $\sigma^\M \not \in \A_i(t,\psi)$. 
				The proposition follows.
			\qed \end{proof}
		
		\subsection{An Improved Tableau}
		By taking a closer look at the interactions between the agents we can further improve the efficiency of our tableau procedure and we do that in this section. We use this improvement to prove an upper bound on the complexity of a general class of logics. We need the following definitions and lemma \ref{lem:clusters}, which is a generalization of a result from \cite{Achilleos:wollic11} and has appeared in simpler forms in \cite{Achilleos14TRtwoagent}.

%\begin{itemize}
%		\item[$Q = $] $(V|_S\cup C|_S)^*$:  if $i Q j$, then $i$'s justifications somehow affect $j$'s justifications; conversely, $j$'s accesibility relations somehow affect $i$'s accessibility relations;
%		\item[
		First we define an additional equivalence relation: 
		$\equiv_{VC} = Q \cap Q^{-1}$. As an equivalence relation, this one too gives equivalence classes on $S$ and they are
%\end{itemize}	
		$P_{VC} = \{\many{P}{k_{VC}}\}$.
		Notice that as $C^* \subseteq (V\cup C)^*$, $\equiv_C \subseteq \equiv_{VC}$ and therefore $P_C$ is a refinement of $P_{VC}$. We'll call $P = P_{VC}$. 
		Furthermore, notice that for any $L \in P$, either $\exists x,y \in L$ s.t. $x V y$, or $L \in P_C$. In the first case, $L$ will be called a \emph{V-class} of agents and in the second case it will be called a \emph{C-class} of agents. For each agent $i \in [n]$, $P(i)$ will be the equivalence class $L \in P$ s.t. $i \in L$ and 
		To help keep the relationship between the sets of equivalence classes in mind, notice that $i \in \chi(i) \subseteq P(i) \subseteq S \subseteq [n]$.
		Furthermore, we can extend relations $<_C$ and $<_{VC}$ on $P$ in the same way they were defined on $P_C$.
%		: $P_1 \leq_C P_2$ iff $\exists x \in P_1 \exists y \in P_2$ s.t. $x C_F^* y$ and $P_1 \leq_{VC} P_2$ iff $\exists x \in P_1 \exists y \in P_2$ s.t. $x Q y$. Also,  $P_1 <_{C} P_2$ iff  $P_1 \leq_{C} P_2$ and  $P_1 \not\leq_{C} P_2$ 
%		and similarly for  $P_1 <_{VC} P_2$. Then, define $P_1 \leq_V P_2$ iff there are $x \in P_1$ and $y \in P_2$ s.t. $x Q V Q y$, that is, there are $ x_1, x_2 \in S$, where $x Q x_1 V x_2 Q y$. $P_1 <_V P_2$ iff $P_1 \leq_V P_2$ and $P_2 \not \leq_V P_1$. 		

		\begin{lemma}\label{lem:clusters}
			Let $\M = (W,(R_i)_{i\in [n]}, (\A_i)_{i\in [n]}, \V)$ be a $(J^n_{D,F,V,C})_\cs$ F-model of at most $2^{|\phi|}$ states, $P_a \in P$ a V-class of agents, and $u \in W$.
			Then, there are states of $W$, $(a_i)_{i \in P_a}$, such that 
			\begin{enumerate} %%%%%%%%%%%%%%%% I MAY NEED TO ADD PROPERTIES HERE...
				\item For any $i \in P_a$, $u R_i a_i$.
				\item For any $i,j \in P_a$, $v, b \in W$, if $a_i, b R_j v$, then $b R_j a_j$.
			\end{enumerate}
			$(a_i)_{i \in P_a}$ will be called \emph{a $P_a$-cluster for $u$}.
		\end{lemma}	
		\begin{proof}	
				For this proof we need to define the following.
				Let $i \in [n]$, $w, v \in W$. \emph{An $E_V$-path ending at $i$} (and starting at $i'$) from $w$ to $v$ is a finite sequence $\many{v}{k+1}$, such that for some $\manyk{j} \in [n]$, $\many{E}{k-1} \in \{C^{-1},V^{-1}\}$, where for some $j \in [k - 1]$ $E_j = V^{-1}$ and $j_k = i$ (and $j_1 = i'$), for every $a \in [k - 1]$, $j_a E_a j_{a + 1}$ and if $E_a = C^{-1}$, then $v_{a+1} = v_{a+2}$, while if $E_a = V^{-1}$, then $v_{a+1} R_{j_{a+1}} v_{a+2}$ and $v_1 = w$, $v_k = v$, $v_1 R_1 v_2$. The $E_V$-path \emph{covers} a set $s\subseteq [n]$ if $\{\manyk{j}\} = s$. For this path and $a \in [k]$, $v_{a+1}$ is a $j_a$-state. Notice that if there is an  $E_V$ path ending at $i$ from $w$ to $v$ and some $j\in s$ and $z\in W$ such that the path covers $s$ and $z R_j w$, it must also be the case that $w,z R_i v$.
																
			Let $p : [m] \To P_a$ be such that $m \in \nat$, $p[[m]] = P_a$ and for every $i+1 \in [m]$, either $p(i+1) C p(i)$ or $p(i+1) V p(i)$ and there is some $i+1 \in [m]$ such that $p(i+1) V p(i)$. For any $s \in W$, $x \in \nat$ let $b_0(s),b_1(s),b_2(s),\ldots,b_{m}(s)$ be the following: $b_0(s) = s$, for all $k \in [m]$, 
			$b_{1}(s)$ will be such that there is an $E_V$ path ending at $p(1)$ from $s$ to $b_{1}(s)$ and covering $P_a$ and if $k>1$, $b_k(s)$ is such that 
			$b_0(s),b_1(s),b_2(s),\ldots,b_{k}(s)$ is an $E_V$ path ending at $p(k)$.
			Let $(b^x_i)_{i \in [m], x\in \nat},(a^x_i)_{i \in [m], x\in \nat}$ be defined in the following way. For every $i \in [m]$,
			$b^0_i = b_i(u)$ and for every $x \in \nat$, 
			$a^x_i = b_i(b_m^x)$. 
			Finally, for $0 < x \in \nat$, 
			$(b^x_i)_{i \in [m]}$ is defined in the following way.
			If there are some $b_x,v \in W$, $i,j \in P_a$, such that $b_x R_{j} v$, $a^{x-1}_i R_{j} v$ and not $b_x R_j a^{x-1}_j$,
			then for all $i \in P_a$, $b_i^x = b_i(v)$. Otherwise, $(b^x_i)_{i \in P_a} = (a^x_i)_{i \in P_a}$.
			By induction on $x$, we can see that for every $x,y \in \nat$, $i \in P_a$, if $y\geq x$, then $b_x R_i b_i^y, a_i^y$. Since the model has a finite number of states, there is some $x \in \nat$ such that for every $y \geq x$, $(b^y_i)_{i \in P_a} = (a^y_i)_{i \in P_a}$.
			Therefore, we can pick appropriate $(a_i)_{i \in P_a}$ among $(a^k_i)_{i \in P_a}$ that satisfy conditions 1, 2.
		\qed \end{proof}

		We recursively define relation $\rightarrow$ on $S^*$: 
		\begin{itemize}
			\item if $jC_F i$ then $i \rightarrow j$;
			\item if $j V i$, then $ij \rightarrow j$;
			\item If $\beta \rightarrow \delta$, then $\alpha \beta \gamma \rightarrow \alpha \delta \gamma$.
		\end{itemize} 
		$\rightarrow^*$ is the reflexive, transitive closure of $\rightarrow$.
		%\end{definition}
		$\rightarrow^*$ tries to capture the closure of the conditions on the accessibility relations of a frame. This is made explicit by observing that		
%		\begin{note}
			if for some frame $(W,(R_i)_{i\in [n]})$, $a R_{i_1}R_{i_2}\cdots R_{i_k}b$ and $i_1i_2\cdots i_k \rightarrow^* j_1j_2\cdots j_l$, then $a R_{j_1}R_{j_2}\cdots R_{j_l}b$. Furthermore, if, in addition, $l = k$, then for every $r \in [k]$, $j_r C_F^* i_r$.
%		\end{note}		
%
		For every agent $i \in S$, we introduce a new agent, $\overline{i}$ and we extend $\rightarrow^*$, so that when $P(i)$ a V-class, then $\overline{i} \rightarrow^* \chi$ for every $\chi \in S^*$ such that $\chi \rightarrow^* i$.
		%	and for every $j,k \in S$, where $k \equiv_{VC} i$, $jk\rightarrow^* k$ iff $j\overline{i} \rightarrow^* \overline{i}$. 
		For each $L \in P_C$, we fix some $i \in P_C$ and $\overline{L} = \overline{i}$. Furthermore, if $xy \in P_C^* \cup S^*$, then $\overline{xy} = \overline{x}$ $\overline{y}$.
		%\end{definition}
		This extended definition of $\rightarrow^*$ tries to capture the closure of the conditions on the accessibility relations of a frame like the ones that will result from a tableau procedure as defined in the following.
		
		%\begin{definition}
		Let 
		%$b$ be a branch, 
		$L \in P$ and $\sigma$ be a finite string of elements from  $P_C$. Then, $L$ is \emph{visible} from $\emptyset.\sigma$ if and only if there is some $\chi(i)\subseteq L$, some $\chi \in P_C^*$ and some $\alpha \in S^*$ such that $\sigma = \tau.\chi(i).\chi$ and $\overline{\chi} \alpha \rightarrow^* i$; $\tau.\chi(i)$ is then called the $L$-view from $\sigma$. Notice that there is a similarity between this definition and the statement of lemma \ref{lem:clusters} - this will be made explicit later on.
		
		Then we adjust the tableau by altering rules TrD and S and introduce rule SVB:

\noindent
		\begin{minipage}[l][1.3cm]{0.4\linewidth}
			\[ 
			\inferrule*[right=TrD]{\sigma\ T\ \term{t}{i}\psi}{  \sigma.\chi(j)\ F\ \bot }  
			\]
		\end{minipage}\hfill
		\begin{minipage}{0.55\linewidth}
			if $i \in S$, $\chi(j) \in M_C(\chi(i))$, $j \notin R$, and $P(j)$ is not a $V$-class visible from $\sigma$;
		\end{minipage}
		\begin{minipage}[l][1.3cm]{0.4\linewidth}
			\[
			\inferrule*[right=S]{\sigma.\chi(j) \ F \ \bot }{\sigma.\chi(j).\chi(i) \ F \ \bot} 
			\]
		\end{minipage}\hfill
		\begin{minipage}{0.55\linewidth}
			if $i \in N(j)$
			and it is not the case that 
			$P(i)$ is a $V$-class visible from $\sigma.\chi(j)$;
		\end{minipage}		
		\begin{minipage}[l][1.4cm]{0.4\linewidth}
			\[
			\inferrule*[right=SVB]{\sigma\ T\ \Box_{i}\psi}{\tau.\chi(i)\ T\ \psi}
			\] 
		\end{minipage}\hfill
		\begin{minipage}{0.55\linewidth}
			if $ P(i)$ a $V$-class, visible from $\sigma$, $\tau.\chi(j)$ is the $P(i)$-view from $\sigma$ and $\tau.\chi(i)$ has already appeared in the tableau.
		\end{minipage}
		
		Then we have to redefine the frame $\F(b)$.
		Let $(R_i)_{i \in [n]}$ be such that for every $i \in [n]$,
		\[R_i = \{(\sigma,\sigma.\chi(i))\in (W(b))^2\} \cup  \{ (w,w)\in (W(b))^2 | i \in F \} \]
		\[ \cup \]
		\[\{(\sigma,\tau.\chi(i))\in (W(b))^2|
		P(i) \mbox{ a }V\mbox{-class, $\tau.\chi(j)$ the $P(i)$-view from }\sigma \}\]
		and $\F(b) = (W(b),(R'_i)_{i\in [n]})$, where $(R'_i)_{i\in [n]}$ is the closure of $(R_i)_{i\in [n]}$ as it was defined before.
				
		Proposition \ref{prp:tableaucomplete} and its proof remain the same. To prove proposition \ref{prp:tableaux} for this version of the rules, follow the same proof, but 	for every $a \in W$ and V-class $L$ fix some $L$-cluster for $a$ and  if $P(i)$ is a $V$-class, then for  $(a_j)_{j\in L}$, the fixed $L$-cluster for $\sigma^\M$,  $(\sigma.\chi(i))^\M = a_i$ - and observe that if $P(i)$ a $V$-class, visible from $\sigma$ and $\tau.\chi(j)$ is the $P(i)$-view from $\sigma$, then in model $\M$ there is some $v$ such that $\sigma^\M, (\tau.\chi(j))^\M R_i v$, which by the definition of clusters in turn means that $\sigma^\M R_j (\tau.\chi(j))^\M$. The remaining proof is the same.
		
		Notice that if for every appearing world-prefix $\sigma.\chi(i)$, $i$ is always in the same V-class $L$, then all prefixes are of the form $\emptyset.\chi(j)$, where $j \in L$. In that case we can simplify the box rules and in particular just ignore rule V and end up with the following result.
				
%		\begin{corollary}\label{cor:easycasesinpi2}
%			When there is some V-class $L$ such that for $L' = \{i \in S| \exists i' \in L\ i' C_F^* i \}$ it is the case that $S\smallsetminus R \subseteq L'$, then $(J^n_{D,F,V,C})_\cs$-satisfiability  is in $\Sigma_2^p$.
%		\end{corollary}
		
		\begin{corollary}\label{cor:easycasesinpi2}
			When there is some V-class $L$ such that  for every $i \in S\smallsetminus R $ there is some $i' \in L$ such that $iC_F^* i'$, then $(J^n_{D,F,V,C})_\cs$-satisfiability  is in $\Sigma_2^p$.
		\end{corollary}

		\section{Complexity Jumps}
		
		In this section we look into some more specific cases of multi-agent justification logics and demonstrate certain jumps in the complexity of the satisfiability problem for these logics. We first revisit the two-agent logics from \cite{Achilleos14TRtwoagent}. Like in the previous sections, we assume our constant specifications are schematic and axiomatically appropriate (and in \P\ for upper bounds).
		
		Our definition here of $(J^n_{D,F,V,C})_\cs$ allows for more two-agent logics than the ones that were studied in \cite{Achilleos14TRtwoagent}. It is not hard, though, to extend those results to all two-agent cases of $(J^n_{D,F,V,C})_\cs$: when there are $\{i,j\} = [2] $, $i \in D \setminus F$, $\emptyset \neq V\subseteq \{(i,j), (j,j) \} $, $(i,j) \in C$, and $(j,i) \notin C$, then $(J^2_{D,F,V,C})_\cs$-satisfiability is \PSPACE-complete; otherwise it is in $\Sigma_2^p$ (see \cite{Achilleos14TRtwoagent}).
		
		We will further examine the following two cases. $\J_1, \J_2$ are defined in the following way. $n_1 = n_2 = 3$; $D_1 = D_2 = \{1,2\}$; $F_1 = F_2  = \emptyset$; $V_1 = \emptyset$, $V_2 = \{(3,3)\}$; $C_1 = C_2  = \{(3,1),(3,2)\}$; finally, for $i \in [2]$, $\J_i = (J^{n_i}_{D_i,F_i,V_i,C_i})_{\cs_i}$, where $\cs_i$ is some axiomatically appropriate and schematic constant specification.
		 
		By an adjustment of the reductions in \cite{Achilleos14DiamondFreeArxiv}, as it was done in \cite{Achilleos14TRtwoagent}, it is not hard to prove that
%		\begin{proposition}\label{prp:exphardness} 
			$\J_{1}$ is \PSPACE -hard and $\J_{2}$ is \EXP -hard.\footnote{$\J_1$ would correspond to what is defined in \cite{Achilleos14TRtwoagent} as $\d_2 \oplus_\subseteq \k$ and $\J_2$ to $\d_2 \oplus_\subseteq \df$. Then we can pick a justification variable $x$ and we can either use the same reductions and substitute $\Box_i$ by $\term{x}{i}$, or we can just translate each  diamond-free fragment to the corresponding justification logic in the same way. It is not hard to see then that the original modal formula behaves exactly the same way as the result of its translation with respect to satisfiability - just consider F-models where always $\A_i(t,\phi) = W$.} 
			Notice that the way we prove $\PSPACE$-hardness for $\J_1$ is different in character from the way we prove the same result for the two-agent logics in \cite{Achilleos14TRtwoagent}. For $\J_1$ we use the way the tableau prefixes for it branch, while for $(\jdf \times_C \jd)_\cs$ the prefixes do not branch, but they increase to exponential size.
%			\end{proposition}
		
		In fact, we can see that $\J_1$ is \PSPACE-complete, while $\J_2$ is \EXP-complete. The respective tableau rules as they turn out for each logic are - notice that neither logic has any $\leq_{C}$-maximal V-classes:
		
		\emph{The rules for $\J_1$ are:}
		\\ \begin{minipage}[l][1.7cm]{0.3\linewidth}
			\[ 
			\inferrule*[right=TrB]{\sigma\ T\ \term{t}{i}\psi}{\sigma\ T\ *_i(t,\psi) \\\\ \sigma\ T\ \Box_i\psi }  
			\]
		\end{minipage}
		\begin{minipage}[l][1.3cm]{0.3\linewidth}
			\[ 
			\inferrule*[right=TrD]{\sigma\ T\ \term{t}{1}\psi}{  \sigma.\{1\}\ F\ \bot }  
			\]
		\end{minipage}\hfill
		\begin{minipage}[l][1.3cm]{0.3\linewidth}
			\[ 
			\inferrule*[right=TrD]{\sigma\ T\ \term{t}{2}\psi}{  \sigma.\{2\}\ F\ \bot }  
			\]
		\end{minipage}\hfill
		\begin{minipage}[l][1.7cm]{0.3\linewidth}
			\[ 
			\inferrule*[right=TrD]{\sigma\ T\ \term{t}{3}\psi}{  \sigma.\{1\}\ F\ \bot \\\\ \sigma.\{2\}\ F\ \bot }  
			\]
		\end{minipage}\hfill
		\begin{minipage}[l][1.3cm]{0.3\linewidth}
			\[ 
			\inferrule*[right=Fa]{\sigma\ F\ \term{t}{i}\psi}{\sigma\ F\ *_i(t,\psi)} 
			\]
		\end{minipage}\hfill
%			\\
		\begin{minipage}[l][2cm]{0.3\linewidth}
			\[
			\inferrule*[right=SB]{\sigma\ T\ \Box_{i}\psi}{\sigma.\{i\}\ T\ \psi} 
			\] 
%		\end{minipage}\hfill
%		\begin{minipage}{0.55\linewidth}
			if $\sigma.\{i\}$ has already appeared;
		\end{minipage}
		\begin{minipage}[l][1.9cm]{0.25\linewidth}
			\[ 
			\inferrule*[right=C]{\sigma\ T\ \Box_{3}\psi}{\sigma\ T\ \Box_{1}\psi \\\\ \sigma\ T\ \Box_{2}\psi}  
			\]
		\end{minipage}\hfill

Notice that the maximum length of a world prefix is at most $|\phi|$, since the depth (nesting of terms) of the formulas decrease whenever we move from $\sigma$ to $\sigma.\{i\}$. Also notice that when we run the $*$-calculus, there is no use for rule $*$V-Dis, so we can simply run the calculus on one world-prefix at the time, without needing the whole frame. Therefore, we can turn the tableau into an alternating polynomial time procedure, which uses a nondeterministic choice when the tableau would make a nondeterministic choice (when we apply the propositional rules) and uses a universal choice to choose whether to increase prefix $\sigma$ to $\sigma.\{1\}$ or to $\sigma.\{2\}$. This means that $\J_1$-satisfiability is \PSPACE-complete.

\emph{The rules for $\J_2$ are the same with the addition of:}
		\\
		\begin{minipage}[l][1.3cm]{0.4\linewidth}
			\[
			\inferrule*[right=V]{\sigma\ T\ \Box_{3}\psi}{\sigma\ T\ \Box_{3}\Box_3\psi} 
			\] 
		\end{minipage}\hfill 
		\begin{minipage}[l][1.5cm]{0.4\linewidth}
			\[
			\inferrule*[right=S]{\sigma.\{i\} \ F \ \bot }{\sigma.\{i\}.\{1\} \ F \ \bot \\\\ \sigma.\{i\}.\{2\} \ F \ \bot} 
			\]
		\end{minipage}\hfill
		
		For the tableau procedure of $\J_2$ we have no such bound on the size of the largest world-prefix, so we cannot have an alternating polynomial time procedure. As before, though, the $*$-calculus does not use rule $*$V-Dis, so again we can run the calculus on one world-prefix at the time. Furthermore, for every prefix $w$, $|\{ w\ a \in b \}|$ is polynomially bounded (observe that we do not need more than two boxes in front of any formula), so in turn we have an alternating polynomial space procedure. Therefore, $\J_2$-satisfiability is \EXP-complete.

\end{document}